\documentclass[nofootinbib,aps,reprint,twocolumn,superscriptaddress,showpacs,prl,nofootinbib,floatfix]{revtex4-1}

\usepackage{hyperref}

\usepackage{amssymb}
\usepackage[figuresright]{rotating}

\usepackage{amsmath}
\usepackage{soul}
\usepackage{cancel} 
  \usepackage{wasysym} 
 \usepackage{mathrsfs} 
 \usepackage{graphicx} 
 \usepackage{blkarray}
 \usepackage{amsfonts} 
 \usepackage{amsbsy} 
 \usepackage{amscd} 
 \usepackage{amssymb}
 \usepackage{mathrsfs}
\usepackage{mathtools}
\usepackage{amsfonts,amstext}
\usepackage{amsmath,amsthm,amssymb}
 \usepackage[toc,page]{appendix}
  \usepackage{bm}
 \usepackage{accents} 
 \usepackage{multirow} 
\usepackage{color}
\usepackage[all]{xy}
\usepackage[colorinlistoftodos]{todonotes}

\newcommand{\be}{\begin{equation}}
\newcommand{\ee}{\end{equation}}
\newcommand{\beq}{\begin{equation*} }
\newcommand{\eeq}{\end{equation*}  }
\newcommand{\bqa}{\begin{eqnarray*} }
\newcommand{\eqa}{\end{eqnarray*}}
\newcommand{\bea}{\begin{eqnarray}}
\newcommand{\eea}{\end{eqnarray}}

  \def \upmns {{\rm{U_{PMNS}}}}
   
  \def \uosc {{\rm{V_{osc}}}}
  \def \vosc {{\rm{V_{osc}}}}
 
   \def \upmnsca {{\rm{\cancel{U}_{PMNS}}}}

\def\beq{\begin{equation}}
\def\eeq{\end{equation}}

\def\beqa{\begin{eqnarray}}
\def\eeqa{\end{eqnarray}}
 
\def\ben{\begin{enumerate}}
\def\een{\end{enumerate}}
\def\bit{\begin{itemize}}
\def\eit{\end{itemize}}

\newcommand{\ket}[1]{|{#1}\rangle}
\newcommand{\N}[1]{\widetilde \nu}
\newcommand{\Nm}[1]{\ket{\widetilde \nu^{(m)}_{#1}}}
\newcommand{\Nf}[1]{\ket{\widetilde\nu^{(f)}_{#1}}}
\newcommand{\nuf}[1]{\ket{\nu^{(f)}_{#1}}}
\newcommand{\num}[1]{\ket{\nu^{(m)}_{#1}}}

 \newtheorem{theorem}{Theorem}
 
 \newtheorem{prop}{Proposition}
 
 \newtheorem{defi}{Definition} 
 
 \newtheorem{cor}{Corollary}
 
 \newtheorem{lem}{Lemma}

\begin{document}

 
 
 \title{
Neutrino mixing, interval matrices and singular values
 }
 
\author{Krzysztof Bielas}
\affiliation{Institute of Physics, University of Silesia,  Katowice, Poland}
\affiliation{Copernicus Center for Interdisciplinary Studies, Szczepa\'{n}ska 1/5, 31-011 Krak\'{o}w, Poland}
\author{Wojciech Flieger}
\author{Janusz Gluza}
\affiliation{Institute of Physics, University of Silesia, Uniwersytecka 4, 40-007 Katowice, Poland}
\author{Marek Gluza}
\affiliation{Dahlem Center for Complex Quantum Systems, Freie Universit{\"a}t Berlin, 14195 Berlin, Germany}
 
\date{\today}
\begin{abstract}
We study the properties of singular values of mixing matrices embedded within an experimentally determined interval matrix. We argue that any physically admissible mixing matrix needs to have the property of being a contraction. This condition constrains the interval matrix, by imposing correlations on its elements and leaving behind only physical mixings that may unveil signs of new physics  in terms of extra neutrino species. We propose a description of the admissible three-dimensional mixing space as a convex hull over experimentally determined unitary mixing matrices parametrized by Euler angles, which allows us to select either unitary or nonunitary mixing matrices. The unitarity-breaking cases are  found through singular values and we construct unitary extensions, yielding a complete theory of minimal dimensionality larger than three through the theory of unitary matrix dilations.
We discuss further applications to the quark sector.

\end{abstract}

\maketitle 
\allowdisplaybreaks

\section{Introduction}

\allowdisplaybreaks
Studies of neutrinos are at the frontier of contemporary research in particle physics.
These fundamental particles crucially influence processes occurring inside the Sun \cite{Bahcall:1976zz}, stars, and supernovae \cite{Giunti:2007ry,Mohapatra:2004}. 
In order to learn about their properties, there are dozens of short- and long-baseline neutrino oscillation experiments studying, e.g., their appearance or disappearance \cite{unb}. 
Thanks to them, we know that at least two out of three neutrinos are massive, though their masses are extremely tiny, at most at the electronvolt level, $m_{\nu}\sim {\cal{O}}(1)$ eV \cite{Patrignani:2016xqp}.
Gathering this information was a highly nontrivial task as neutrino experiments involve the challenge of low event statistics.
Among unsolved and important problems in neutrino physics remains the issue of the total number of neutrino species in nature. 
Do we really have only the three electron, muon, and tau neutrino \emph{flavors} as prescribed by the neutrino theory of the Standard Model (SM) \cite{Giunti:2007ry,Mohapatra:2004}? 
This knowledge is
 of paramount importance for progress in understanding particle physics and theory beyond the Standard Model (BSM), but also in the astrophysics and cosmology of the big bang, leptogenesis and baryogenesis, and dark matter  \cite{Giunti:2007ry,Mohapatra:2004,Lesgourgues-Mangano-Miele-Pastor-2013,Rasmussen:2017ert}. 
 The point is that additional neutrino species are likely massive, affecting the dynamics of many processes and systems, including the Universe as whole. Their existence is theoretically appealing as they could provide an explanation of the smallness of masses of known neutrinos, e.g. by the celebrated see-saw mechanism \cite{Minkowski:1977sc,Yanagida:1979as,GellMann:1980vs,Mohapatra:1979ia}.
 There is currently no compelling experimental evidence for \emph{extra} neutrino states, despite  direct collider \cite{Achard:2001qv,Blondel:2014bra,Golling:2016gvc,Oh:2016ead}
  and indirect electroweak precision studies \cite{Nardi:1991rg,Nardi:1994iv,Kniehl:1994vk,Illana:2000ic,delAguila:2008pw,Fernandez-Martinez:2016lgt} providing bounds on their masses and couplings. 
 As a dim clue for their presence one may  consider an outcome of  the Large Electron-Positron collider (LEP) studies where the central value for the effective number of light neutrinos $N_\nu$ was determined by analyzing around 20 million  $Z$-boson decays, yielding $N_\nu=2.9840 \pm 0.0082$ \cite{ALEPH:2005ab,Novikov:1999af}. 
In fact, a natural extension of the SM by right-handed, sterile  neutrinos leads to an $N_\nu$  value less than three \cite{Jarlskog:1990kt}.
 There are also intensive studies concerning sterile eV-scale neutrinos, connected with oscillation phenomena. In the Liquid Scintillator Neutrino Detector (LSND) experiment an excess of electron antineutrinos appearing in a mostly muon antineutrino beam at the $3.8\sigma$ level was observed while the SM would predict no significant effect \cite{Aguilar:2001ty}. To solve this puzzle conclusively new experiments are under way \cite{Gariazzo:2017fdh}.  For recent MiniBoone results, see \cite{Aguilar-Arevalo:2018gpe}.
 The question of whether sterile neutrinos exist is being researched by ongoing studies performing global analyses of neutrino oscillation data \cite{Gariazzo:2017fdh, Dentler:2018sju}.
\\
  \indent
In the description of phenomena like neutrino oscillations,  \emph{mixing matrices} are the central objects.
In the SM scenario with three neutrino species, the mixing matrix is known as the Pontecorvo-Maki-Nakagawa-Sakata matrix (PMNS) \cite{Pontecorvo:1957qd,Maki:1962mu}. It is three dimensional and unitary, and it can be parametrized by Euler angles.
When the evaluation of experiments is performed, the hope from the BSM perspective is that an inconsistency in data analysis -- in particular, \emph{violation of unitarity} of the mixing matrix -- would give a hint for the existence of new neutrino states.
In this work, assisted by concepts and theorems taken from matrix theory and convex analysis \cite{books/daglib/0019187,Allen06matrixdilations,Halmos:1950,Sz.-Nagy:1953,magaril2003convex,Krantz2014} we describe  an elegant approach to mixing phenomena capable of capturing SM and BSM within the same framework.

At the foundation of our studies lies the analysis of singular values of mixing matrices in the form of an interval matrix  which gathers knowledge of experimental errors. 
Firstly, we characterize physical mixing matrices by looking at the largest singular value of a given mixing matrix (which equals the operator norm) and derive on physical grounds that it must be less than or equal to unity, a matrix property known as \emph{contraction}. 
Using the  notion of contractions we consistently stay within the region of physical states with properly correlated mixing elements. 
Secondly, we study unitarity violation as witnessed by any of the singular values being strictly less than one, which has a direct physical consequence and means that the three SM neutrinos mix with unknown species.
Therefore identifying such a situation is a smoking gun signal for the existence of additional neutrinos.
Finally, we employ the theory of \emph{unitary matrix dilations} in order to find a unitary extension of any three-dimensional mixing matrix which is physically admissible yet not unitary. 
We apply this method to an example from experimental data and discuss how this approach can be used to find a minimal number of necessary extra neutrino states in a BSM scenario, leading to a complete theory based on a higher-dimensional unitary mixing matrix.

\section{Setting}
\label{sec:setting}
We begin our discussion with the SM scenario of three  weak flavor -- electron, muon, tau -- neutrinos. 
In this framework, mixing of neutrinos is modeled by single-particle asymptotically free scattering states with a given momentum and spin which are emitted in a fixed \emph{flavor} state $\nuf e, \nuf \mu$, or $\nuf \tau$ and then mix coherently between different \emph{mass} states $\num 1, \num 2$, and $\num 3$, defined by \cite{Giunti:2007ry}
\begin{equation}
\nuf \alpha =\sum_{i=1}^3
\left( \upmns  \right)_{\alpha i} 
{\num i}\;.
\label{class}
\end{equation}
The PMNS \emph{mixing matrix} $\upmns$ is unitary and can be parametrized by   \cite{Maki:1962mu, Kobayashi:1973fv, Bilenky:1987ty}
{\small{\begin{eqnarray}
  \label{eq:U3m}
 && \upmns = \nonumber 
\\ && \hspace*{-.7cm}  \begin{pmatrix}
    1 & 0 & 0 \\
    0 & c_{23}  & {s_{23}} \\
    0 & -s_{23} & {c_{23}}
  \end{pmatrix}
  \begin{pmatrix}
    c_{13} & 0 & s_{13} e^{-i\delta_\text{}} \\
    0 & 1 & 0 \\
    -s_{13} e^{i\delta_\text{}} & 0 & c_{13}
  \end{pmatrix}
  \begin{pmatrix}
    c_{12} & s_{12} & 0 \\
    -s_{12} & c_{12} & 0 \\
    0 & 0 & 1
  \end{pmatrix}, \nonumber \\
\label{upmns}
\end{eqnarray}
}}
\!\!where we denote $c_{ij} \equiv \cos( \theta_{ij})$, $s_{ij} \equiv \sin(\theta_{ij})$,
and the Euler rotation angles $\theta_{ij}$ can be taken without loss
of generality from the first quadrant, $\theta_{ij} \in [0, \pi/2]$, and the $CP$ phase $\delta \in [0, 2\pi]$.
The current global $3\nu$ oscillation analysis \cite{Esteban:2016qun,nufit} gives at $3\sigma$ C.L. 
\begin{eqnarray}
 \theta_{12} \in [31.38^\circ, 35.99^\circ ]\ , \quad 
 \theta_{23} \in [38.4^\circ, 53.0^\circ ]\ , \nonumber\\
 \theta_{13}  \in  [7.99^\circ, 8.91^\circ ]\ ,  \quad\quad\text{and}\quad\quad 
  \delta  \in   [0, 2 \pi]\;.
 \label{exp3s}
 \end{eqnarray}
These results are  independent of the normal or inverse mass hierarchies \cite{Smirnov:2013cqa,Olive:2016xmw}, which is not of first concern in this work.
The exact ranges can differ also slightly in other analyses \cite{Capozzi:2016rtj, Forero:2014bxa}.

In the above,  it was assumed that mixing among light and active neutrino states is complete -- hence the neutrino mixing matrix is unitary. 
However, the situation can be more complicated. 
In a BSM scenario other neutrino mass and flavor  states  can be present that we denote by $\Nm j$ and $\Nf j$ for $j=1,\ldots, n_R$, respectively.
In this scenario mixing between an extended set of neutrino mass states $\{\num \alpha,\Nm \beta\}$ with flavor states $\{\nuf \alpha,\Nf \beta\}$ is described by
\begin{align}
\begin{pmatrix} { \nuf \alpha} \\ 
 \Nf \beta\end{pmatrix}  &=
 \begin{pmatrix} {{ V }} & V_{lh} \\  V_{hl} & V_{hh}
  \end{pmatrix} 
\begin{pmatrix} {\num \alpha } \\ 
\Nm \beta\end{pmatrix} 
\equiv U  \begin{pmatrix} {{ \num \alpha }} \\ \Nm \beta\end{pmatrix}\;.
\label{ugen}
\end{align}
{Such block structures of the unitary $U$ are present in many neutrino mass theories. Note that (\ref{ugen}) effectively implements an assumption of unitary mixing restricted to the level of single-particle states only, e.g., neglecting interaction effects which are expected to be weak. 
Indices "$l$" and "$h$" in \eqref{ugen} stand for "light" and "heavy" as usually we expect extra neutrino species to be much heavier than known neutrinos; cf. the see-saw mechanism  \cite{Minkowski:1977sc,Yanagida:1979as,GellMann:1980vs,Mohapatra:1979ia}. 
However, it does not have to be the case: they can also include light sterile neutrinos, which effectively decouple in weak interactions, but are light enough to be in quantum superposition with three SM active neutrino states and to take part in the neutrino oscillation phenomenon \cite{Capozzi:2016vac}. 
} 

The observable part of the above is the transformation from mass $\num \alpha, \Nm \beta$ to SM flavor $\nuf \alpha$ states and reads 
\begin{eqnarray}
\nuf \alpha &=& 
\sum_{i=1}^3\underbrace{\left( V  \right)_{\alpha i} {\num i}}_{\rm SM \;part}
+ 
\sum_{j=1}^{n_R}\underbrace{
{ \left(V_{lh}  \right)}_{\alpha j} \Nm j}
_{\rm BSM \; part}\;. \label{gen}
\end{eqnarray}  
If $V$ is not unitary then  there necessarily is a light-heavy neutrino "coupling" and the mixing between sectors is nontrivial $V_{lh} \neq 0 \neq V_{hl}$. 
Without extra states $\Nm ~$, we end up with the situation described in \eqref{class}-\eqref{exp3s}, $V \to \upmns$; i.e., there
 are either no BSM neutrinos or they are decoupled on the level of the joint mixing matrix.

\section{Physically admissible mixing matrices are contractions}
In this section we will make precise the notion {of physically admissible mixing matrices}.
To this end, we will study the \emph{singular values} $\sigma_i(V)$ of a given matrix $V$, which are equal to the positive square roots of the eigenvalues $\lambda_i$ of the matrix $VV^\dagger$; i.e., $\sigma_{i}(V)\equiv \sqrt{\lambda_{i}(VV^{\dag})}$ for $i=1,2,3$ \cite{books/daglib/0019187}.
Singular values generalize eigenvalues to all kinds of matrices, e.g., those not diagonalizable by a similarity transformation or even rectangular ones, and have useful properties that in particular can be related to the operator norm $\|V\|\equiv \max_i \sigma_i(V)$. 
In the SM scenario one would only consider unitary matrices; hence $\|V\|=1$ and all singular values are equal (see Appendix A).  
In this work, we are also interested in constraints on $V$ as a principal \textit{submatrix} of a unitary $U$ realizing some BSM scenario \eqref{ugen}. 
For any such matrix $V$, the operator norm is bounded by unity
\begin{equation}
 \Vert V \Vert \le 1\;,
  \label{eq:contr}
 \end{equation}
a matrix property known as \textit{contraction}.
In other words, if $U$ is unitary, then $\Vert U\Vert=1$ and for any submatrix $V$ of $U$ it holds that $\Vert V\Vert\leq \Vert U\Vert=1$; see Appendix A for a simple proof.
Observe that $\Vert V \Vert=1$ is not sufficient for $V$ to be unitary and any significant deviation of \emph{any} singular value from unity $\sigma_i(V)< 1$  signals BSM physics.
  Physically, measuring a mixing matrix with nonunit singular values means that a given neutrino mixes with other ones that are not being observed and hence the unitarity loss.
 Note that any observable  mixing matrix must be a contraction both in the SM and BSM scenario. 
Moreover, singular values are suitable quantities while working with experimental data, since they are stable under the addition of perturbing error matrices and the resulting errors of the operator norm can be upper bounded, while the stability of eigenvalues can be in general very weak, e.g., violating Lipschitz continuity \cite{books/daglib/0019187}. It can be achieved only if matrices after the perturbation remain normal \cite{books/daglib/0019187}, a condition that obviously cannot be fulfilled generally when considering experimental data.

In this work, we show how the contraction property allows us to distinguish physically admissible mixing matrices. 
Namely, recall that if $V$ is a submatrix of some larger unitary mixing matrix $U$, then it must be a contraction.
Conversely, as we will show presently any $V$ which is a contraction can be completed into a unitary mixing matrix $U$ whose minimal dimension can be read off from the singular values of $V$.
Hence we establish the following characterization useful for data analysis, allowing us to decide whether a candidate mixing matrix $V$ is physically admissible.
\begin{defi}[Physically admissible mixing matrix]
A matrix $V$ is a physically admissible mixing matrix if and only if it is a contraction, i.e., $\|V\|\le 1$.
\end{defi}

\section{Interval matrices, unitarity violation, and contractions}
Though the matrix $\upmns$ is unitary, information on BSM physics can be  hidden there. 
To see this one should ask what would be the result of a fit assuming unitarity in the case that the mixing was actually nonunitary?
In a BSM scenario  \eqref{gen}, a unitary fit to (\ref{upmns}) would hide the BSM physics in the error bars and hence the experimental Euler angle ranges may reflect not only measurement inaccuracies but also the hypothetical nonunitarity of the underlying mixing.
For similar reasons, the search for BSM via unitary triangle analysis \cite{Jarlskog:1988ii,Xing:2015wzz} is based on PMNS data. So far experimental analysis is not precise enough to confirm or exclude definitively BSM \cite{Patrignani:2016xqp}.

In order to find the nonunitary cases, we discretize the experimentally allowed ranges in \eqref{exp3s}, calculate the corresponding $\upmns$ matrices using \eqref{upmns}, and collect the extreme values of each matrix element that occurred into an interval matrix $\vosc$ \cite{Esteban:2016qun}
\vspace{-10pt}

{\small{\begin{eqnarray}
  \label{ranges}
 && \hspace*{-.5cm}\upmns \to \vosc = \\
&& \begin{pmatrix}
    0.799 \div 0.845 & 
    0.514 \div 0.582 &
    0.139 \div 0.155
    \\
    -0.538 \div -0.408 &
    0.414 \div 0.624 &
    0.615 \div 0.791
    \\
    0.22 \div 0.402 &
    -0.73 \div -0.567 &
    0.595 \div 0.776
    \end{pmatrix}\;. \nonumber
\end{eqnarray}}}
We will write $V\in \vosc$ whenever all entries of $V$ lie in the intervals of \eqref{ranges}.
This interval matrix is real as we have fixed for simplicity in \eqref{exp3s} the CP-phase $\delta$ to be zero, but our analysis can also be applied to complex mixing matrices.
The exact values in this interval matrix can differ slightly depending on global fits and considered approaches \cite{Esteban:2016qun,Capozzi:2016rtj, Forero:2014bxa,Parke:2015goa}.
Our construction of the  interval matrix is based on \cite{Esteban:2016qun,nufit}, {where the interval matrix was obtained in the same way, i.e., by looking at the extreme values of the entries of the mixing matrices $V_{ij}$ for all possible Euler angles consistent with the oscillation data. As an alternative, we will also refine this procedure by looking at convex combinations of $\upmns$ matrices which should be even closer to the data by retaining correlations between matrix elements.
In particular, this allows us to construct candidate BSM matrices as toy examples to study various methods related on mixing matrices close to the data.

It is important to observe that it is not necessary to construct (\ref{ranges}) from $\upmns$. In principle, such an interval matrix could be derived directly from experimental data. 
In the neutrino sector, direct experimental access to each of the entries of the $3\times3$ matrix individually is presently not possible and experimental analyses based on $\upmns$ are a natural choice. 
If this were possible, then the interval matrix would become a useful way of bringing various experimental findings together. Indeed, approaches to oscillation analysis independent of PMNS are possible \cite{Parke:2015goa}.

Typically, unitarity violation of a neutrino mixing matrix $\upmnsca$, where the slash emphasizes unitarity-breaking, can be parametrized in various ways
\cite{Bekman:2002zk,FernandezMartinez:2007ms,Antusch:2006vwa,delAguila:2008pw,Goswami:2008mi,Rodejohann:2009cq,Antusch:2009pm,Fong:2016yyh,Blennow:2016jkn,Dutta:2016czj}.
A frequent approach is based on the polar decomposition, introducing a unitary $Q$ and a Hermitian matrix $\eta$, to write \cite{FernandezMartinez:2007ms,Antusch:2006vwa}
\begin{equation}
{\upmnsca} {{{\rm{\cancel{U}^{\dagger}_{PMNS}}}}}= [(I+\eta)Q] [(I+\eta)Q]^\dagger \equiv I+\varepsilon \;.\label{upmnseps}
\end{equation} 
Here, $\eta_{ij}$ and $\varepsilon_{ij}$ "measure" how far from unitary the PMNS matrix can be. 
There is also another commonly used parametrization known as the
$\alpha$ parametrization \cite{Xing:2007zj,Xing:2011ur}:
\begin{equation}
{\upmnsca} 
\equiv (I-\alpha)U. \label{upmnsalpha}
\end{equation}
Here $\alpha$ is a lower triangular matrix and $U$ is unitary. Such a triangular structure of $\alpha$ is especially convenient for singular values analysis \cite{porwit}. 
This parametrization is often used in oscillation analysis, e.g., \cite{Escrihuela:2015wra,Blennow:2016jkn,Escrihuela:2016ube}. 
We discuss both $\alpha$ and $\eta$ parametrizations in the wider context of matrix analysis in Appendix~E.
 
Observe that decompositions given only by \eqref{upmnseps} or \eqref{upmnsalpha} need some extra conditions to produce contractions exclusively, as in general it can happen that although a given matrix  lies within experimental limits and is of the form given by Eq. \eqref{upmnseps} or \eqref{upmnsalpha}, it will not be a contraction (for a proof, see Proposition \ref{prop1} below). 
In particular, such a condition can be provided by embedding a three-dimensional mixing matrix into a larger unitary one. 
Accordingly, it is standard in the  neutrino unitarity-breaking literature to take one of the approaches \eqref{upmnseps} or \eqref{upmnsalpha} \textit{together} with that embedding as the precise definition of the so-called $\alpha$ or $\eta$ parametrization (cf. \cite{Antusch:2006vwa, FernandezMartinez:2007ms, Xing:2007zj, Antusch:2009pm, Blennow:2011vn, Xing:2011ur, Escrihuela:2015wra, Blennow:2016jkn,Agostinho:2017wfs}). 
Therefore, by combination with such additional conditions, the contraction property of the mixing matrix is secured; see Appendix E for further discussion.
However, it is common to present the data of such analyses in the form of an interval
matrix, where the correlations between elements are lost. 
If one would like to consider a mixing matrix $\eta$ taken from such an interval matrix as a point of departure, it is  profitable to find a condition on that particular $\eta$ matrix to be physical, i.e., to give rise to a physically admissible mixing matrix.
For this, Proposition \ref{prop2} characterizes a particular sufficient condition securing any candidate mixing matrix taken from an interval matrix to be physically admissible.
A similar argument can be proven for the $\alpha$ parametrization  \eqref{upmnsalpha} and, in fact, it has been shown in \cite{Blennow:2016jkn} that (\ref{upmnseps}) and (\ref{upmnsalpha}) are equivalent. Therefore only the $\eta$ case is regarded in the following.

\begin{prop}
\label{prop1}
Let $\epsilon_{ij}>0$ and consider a set of possibly nonunitary matrices 
\begin{equation}
\cancel{\Theta} = \{\ V = (1+\eta) U \ \mathrm{:}\ \eta=\eta^\dagger, |\eta_{ij}|\le \epsilon_{ij},\, UU^\dagger= I\ \}\ . \label{eqprop1}\end{equation}
Then $\cancel{\Theta}$ contains matrices which are not contractions and hence are not valid mixing matrices, i.e., are unphysical.
\end{prop}
\begin{proof}
Let $V=(1+\eta) U\in \cancel{\Theta}$ satisfy $\eta\neq0$.
If $\eta$ has a positive eigenvalue $\lambda^+>0$, then we use that $I+\eta$ is diagonalizable and obtain $\|V\|=\|I+\eta\|  \ge 1+\lambda^+>1$ using unitary invariance of the operator norm, and the lower bound comes from the fact that there may be other eigenvalues that are still larger than $\lambda^+$.
If $\eta\preceq 0$, i.e., it has no positive eigenvalues, then we observe that $\tilde \eta = -\eta$ has at least one positive eigenvalue and the constraints of $\cancel{\Theta}$ are satisfied.
Thus we find for $\tilde V=(1+\tilde \eta)U\in\cancel{\Theta}$ that $\|\tilde V\|>1$.
\end{proof}

As an example, let us consider an interval matrix for $\eta$ (see \cite{FernandezMartinez:2007ms}) and its particular elements:
\begin{equation} \label{etamax}
\eta_{max}= 
\left(
\begin{array}{ccc}
0.0054 & 0.000034 & 0.0079 \\
0.000034 & 0.0049 & 0.005 \\
0.0079 & 0.005 & 0.0049
\end{array}
\right)\ .
\end{equation}
The subscript "max" indicates that we have chosen elements of $\eta$ to have largest absolute values given the constraints of the respective interval matrix.

This matrix is Hermitian, as given in (\ref{eqprop1}). As follows from Proposition~\ref{prop1}, more stringent limits, e.g., \cite{Agostinho:2017wfs,Antusch:2014woa}, 
do not change the situation.
Due to the fact that we bound only absolute values, we consider the following two cases:
\begin{equation}
\begin{split}
a) \ I-\eta_{max}, \\
b) \ I+\eta_{max}.
\end{split}
\end{equation}
Performing a singular value decomposition\cite{books/daglib/0019187}, we obtain the following singular values:
\begin{equation}
\begin{split}
a) \ \lbrace 1.00426, 0.995, 0.986 \rbrace, \\
b) \ \lbrace 1.01445, 1.005, 0.996 \rbrace.
\end{split}
\end{equation}
We can see that both spectra, which correspond to the singular values of the matrices $V=(I\pm\eta_{max})U$, contain eigenvalues larger than one, which means that mixing matrices $V$ constructed using these particular matrices are not contractions. 
We consider a general form of the $\eta$ parametrization, where $U \equiv Q$ in (\ref{upmnseps}) is an arbitrary unitary matrix. 
Observe that, regarding our analysis of $\eta_{max}$ taken from Eq.(11), there is a subtle detail. To check whether a matrix $V$ is a contraction, we do not use the unitary matrix $U$ at all.
This follows from the nature of the $\eta$ parametrization which is in fact a polar decomposition. 
The contraction property is based on the operator norm, which is unitarily invariant, which means that only the polar matrix contributes and the unitary part (by definition) does not change the norm.
Thus for the analysis of singular values that we have done, the unitary part is irrelevant.
It should be no surprise that such a particular element $\eta_{max}$ could be unphysical in spite of its Hermiticity, since the very construction of interval matrices destroys correlations between elements, as discussed above. Nevertheless, we can restrict ourselves to physical matrices which are contractions by the following proposition.

\begin{prop}
\label{prop2}
If all $\epsilon_{ij} \le \epsilon$ are sufficiently small, then restricting to negative semidefinite perturbations $\eta\preceq 0$ yields exclusively physically admissible mixing matrices: 
\begin{equation}
\label{eq:theta_good}
\Theta = \{\ V = (1+\eta) U\ \mathrm{:} \ \eta=\eta^{\dag}, \eta\preceq 0, \|\eta\|\le1, \, UU^\dagger= I\ \} \ .
\end{equation}
For any $V=(1+\eta) U\in \Theta$,
we have that the norm of $V$ can be obtained from the largest eigenvalue of the diagonalizable matrix $1+\eta$.
\end{prop}

\begin{proof}
It suffices that $\epsilon<\tfrac 1 n $, where $n$ is the dimension of the matrix $\eta$, and we will find  $\|\eta\|<1$.
As the identity $I$ and $\eta$ are simultaneously diagonalizable and all eigenvalues $\lambda_i(\eta)$ of $\eta$ are nonpositive, we find that $\|V\| = \|I+\eta\|=1+\max_i \lambda_i(\eta)\le 1$, so all $V\in \Theta$ are contractions and thus are admissible mixing matrices.
\end{proof}

We sum up this section in the following way. 
As mentioned already, there are parametrizations which allow us to generate $\upmns$-like $3\times 3$ matrices which  by \emph{construction} are contractions \cite{Antusch:2006vwa, FernandezMartinez:2007ms, Xing:2007zj, Antusch:2009pm, Blennow:2011vn, Xing:2011ur, Escrihuela:2015wra, Blennow:2016jkn,Agostinho:2017wfs}, respecting the present experimental bounds. If not secured directly, the condition of negative semidefinite perturbation (\ref{eq:theta_good}) can be used to ensure that the considered mixing matrices are physically admissible when working particularly with \eqref{upmnseps}. 
In general, it is numerically efficient to check directly the contraction property  \eqref{eq:contr} of examined mixing matrices for any parametrization.

\section{Physical mixing space from Euler angles }
We proceed by characterizing physical mixing matrices consistent with the experimental data.
Firstly, let us note that the set of all (unrestricted) contractions $B=\{ V\in M_{3\times 3}(\mathbb{C}) \;|\; \|V\|\le 1\}$ is a unit ball in operator norm and hence is convex.
This abstract property allows us to describe $B$ in terms of its extreme points which in the case of contractions are unitary matrices $U_{3 \times 3}$ \cite{Stoer:1964:CLU:2715875.2716107}.
In fact, we can easily find that a convex combination  $V=  \sum_{i=1}^{M} \alpha_{i}U_{i}$ of unitary matrices $U_i$ with  $\alpha_{1},...,\alpha_{M} \geq 0$ and $ \sum_{i=1}^{M}\alpha_{i}=1$ is a contraction because $\|V\|\le \sum_{i=1}^{M} \alpha_i \|U_i\|=1$ by the triangle inequality.
For such combinations, restricted to experimentally determined $\upmns$ unitary matrices, we have $V\in \uosc$ because the interval matrix is constructed from extreme values of $U_i$ and convex combinations cannot change these bounds. 
Conversely, when $V\in \uosc$ is a contraction but cannot be written as a convex combination of unitary matrices within allowed angle ranges \eqref{exp3s}, then it means that the construction of $\uosc$ through extreme matrix elements simply introduces discrepancies with the data by disregarding correlations between matrix elements.
Therefore, the set of all finite convex combinations of PMNS matrices given by
\begin{eqnarray}
\Omega 
&\coloneqq&\lbrace  \sum_{i=1}^{M} \alpha_{i}U_{i} \mid U_{i} \in U_{3 \times 3}, 
\alpha_{1},...,\alpha_{M} \geq 0, \sum_{i=1}^{M}\alpha_{i}=1, \nonumber \\
& & \quad \theta_{12}, \theta_{13}, \theta_{23} \ \begin{rm}and\end{rm} \ \delta \ \begin{rm}given \ by\end{rm} \ (\ref{exp3s})   \rbrace
\label{eq:conv}
\end{eqnarray}
comprises all contractions spanned by the experimental data; see Fig.\ \ref{scheme}. This definition takes into account possible nonzero values of the CP-phase $\delta$.

\begin{figure}[t!]
\includegraphics[width=0.9\linewidth]{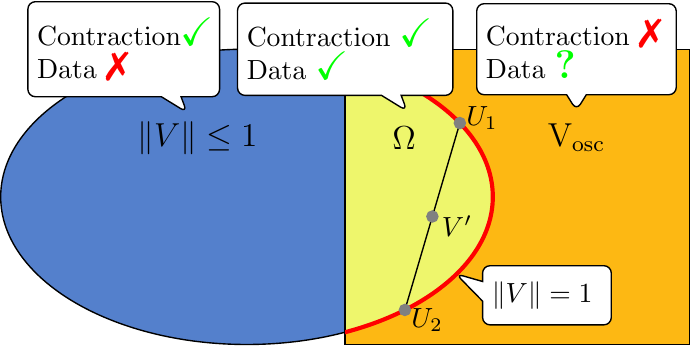} 
\caption{Illustration of the neutrino mixing space. Equation.~\eqref{eq:contr} states that physical mixing matrices $V\in \uosc$ lie within an abstract operator norm unit ball  represented by the ellipse. On the left are cases that are physically admissible, but are excluded by the experimental data \eqref{exp3s}. The middle region $\Omega$ represents relevant mixing matrices consistent with the experiment which are convex combinations  of unitary PMNS matrices. 
The cord slicing $\Omega$ consists of convex combinations of two PMNS matrices $U_1$ and $U_2$, e.g., $V'=\tfrac12U_1+\tfrac12 U_2$, which is further discussed in Sec. VI.
The rectangle on the right depicts the interval matrix form of the data $\uosc$ that is largely unphysical and may include contractions spanned by unitaries outside of \eqref{exp3s}.
}
\label{scheme}
\end{figure}

Currently, it is not possible to measure experimentally values of all elements of the neutrino mixing matrix in the three-dimensional flavor space \cite{Parke:2015goa}. To determine missing elements one uses Euler angles obtained from available data and calculates unreachable matrix elements of the neutrino mixing matrix by \eqref{upmns}.
The set $\Omega$ could be explored in the future in a broader context for data analysis and independent cross-checks with experiments that measure entries of the $3\times3$ mixing matrix directly rather than through Euler angles. 
The matrices in $\Omega$ with $M=1$ yield admissible PMNS matrices, while taking $M\geq 2$ allows us to obtain nonunitary contractions. Although the upper limit $M$ in \eqref{eq:conv} is not unique, in principle it can be bounded from above by Carath\'{e}odory's theorem, which states that if a point $x\in\mathbb{R}^n$ lies in the convex hull of some set $X$, then $x$ can be written as a convex combination of $s$-many points from $X$ such that $s\leq n+1$ \cite{eggleston1958}. Since matrices under study (in the CP-invariant case) are points in $\mathbb{R}^9$, elements of \eqref{eq:conv} are narrowed down to combinations of at most $M=10$ unitary $U_{PMNS}$ matrices. Thus one obtains an upper bound for the number of free parameters under study in this approach.
From the point of view of particle physics phenomenology (but also optimization theory), it would be interesting to refine $M$ even further and look for the smallest possible $M$, called the Carath\'{e}odory number, that would allow us to span $\Omega$ (see, e.g., \cite{Barany2012}). 
While certainly important, this issue goes beyond the present study. 
In the next section we give an example of two unitary PMNS matrices $U_1, U_2$ for which  $V'=\tfrac12U_1+\tfrac12 U_2$ is a contraction $\|V'\|=\sigma_1(V')=1$, but $\sigma_2(V')\approx\sigma_3(V') < 1$ within the accuracy of the interval matrix.
This exemplifies how to find nonunitary BSM cases within admissible set $\Omega$, through the analysis of singular values. 

\section{Dilations: Minimal dimensionality of the complete neutrino mixing matrix }
For BSM mixing matrices, it is possible to find minimal model extensions consistent with the data again using singular values.
A \emph{unitary dilation} is an operation that extends a matrix which is a contraction to a unitary matrix of an appropriate dimension.
Our approach to find a unitary dilation of possible smallest dimension employs the special case of {\it{cosine-sine}} (CS)   \cite{Allen06matrixdilations} decomposition of unitary matrices as follows.
It can be proven that any unitary matrix $U \in M_{(n+m) \times (n+m)}(\mathbb{C})$ can be brought to a canonical form 
$ W^\dagger U Q=
\left(
\begin{array}{cc|c}
I_{r} & 0 & 0 \\
0 & C & -S \\ \hline
0 & S & C
\end{array}
\right)$,
with $r=n-m$ and $C^2+S^2=I_{m}$, where one can choose block-diagonal unitaries $W=W_1\oplus W_2$ and $Q=Q_1\oplus Q_2$.
We use this result to extend any contraction $V \in \Omega$ to a unitary matrix.
First, we find a \emph{singular value} decomposition of $V$, i.e.,
$V=W_{1}\Sigma Q_{1}^{\dag}$,
where $W_{1},Q_{1}$ are unitary, and $\Sigma$ comprises the singular values $\sigma_i(V)$ and is diagonal.
Next, we determine the number $r$ of unit singular values defining $I_{r}$ and collect the rest into a diagonal matrix $C$. 
This yields $\Sigma= I_{r} \oplus C$.
Finally, we define $S=\sqrt{I_{m}-C^2}$  and choose $W_2$, $Q_2$ to be arbitrary unitaries of appropriate dimension.
Conjugating the CS matrix constructed in that way by $W$ and $Q$ yields the unitary dilation $U$ of $V$.
Below an example of a nonunitary contraction $V\in\Omega$ with $m=2$ will be discussed, extended into a unitary matrix $U$ of dimension 5. Any larger unitary dilation of $V$ can by obtained by the general form of CS decomposition; see Theorem 3 in Appendix D.
There, we also prove that $m$, also known as the dimension of the defect space, is the minimal number of new neutrino species necessary to ensure unitarity.
To obtain this number, one thus has to take experimental errors into account.
Assuming that the data $V$ include an error matrix $E$ and are of the form $V+ E$, we can establish the stability of the defect space.
We use Weyl inequalities \cite{Weyl1912,rao} for decreasingly ordered pairs of singular values of $V$ and $V+E$, which read
\begin{equation}
\label{weyl}
 \vert\sigma_i(V+E)-\sigma_i(V)\vert\leq\Vert E\Vert\;.
\end{equation}
In our case $E$ should be taken as $E_{ij}\sim 0.001$ by \eqref{ranges}, and the uncertainty in the precise value of singular values is bounded by $\Vert E\Vert = 0.003$. Note that this criterion applies both to the selection of contractions from the full interval matrix \eqref{ranges} and to determination of the minimal dimension of matrix dilation.

To show the dilation procedure in action, we restrict all matrix elements to real numbers; hence the complex phase $\delta$ is equal to zero and thus we work with orthogonal matrices. The first step is to pick a contraction from the convex hull $\Omega$ (\ref{eq:conv}). 
As an example, let us consider two unitary matrices obtained from the experimental ranges \eqref{exp3s}
\begin{equation} 
\begin{split}
\theta_{12}=31.38^{\circ} , \theta_{23}=38.4^{\circ}, \theta_{13}=7.99^{\circ}, \\
U_{1}=
\left(
\begin{array}{ccc}
0.845 & 0.516 & 0.139 \\
-0.482 & 0.624 & 0.615 \\
0.230 & -0.587 & 0.776
\end{array}
\right)
\end{split}
\end{equation}
\begin{equation} 
\begin{split}
\theta_{12}=35.99^{\circ}, \theta_{23}=52.8^{\circ}, \theta_{13}=8.90^{\circ}, \\
U_{2}=
\left(
\begin{array}{ccc}
0.799 & 0.581 & 0.155 \\
-0.455 & 0.417 & 0.787 \\
0.392 & -0.699 & 0.597
\end{array}
\right).
\end{split}
\end{equation}
The chosen convex combination will be constructed as a sum with an equal contribution of the above matrices:
\begin{equation}
\begin{split}
V'=\frac{1}{2}U_{1}+\frac{1}{2}U_{2}=
\left(
\begin{array}{ccc}
0.822 & 0.549 & 0.147 \\
-0.469 & 0.521 & 0.701 \\
0.311 & -0.643 & 0.687
\end{array}
\right).
\end{split}
\end{equation}
In order to make use of the CS decomposition and parametrize the unitary dilation $U$ of the matrix $U_{11}\equiv V'$, first we have to find its singular value decomposition
\begin{equation}
V'=W_{1}\Sigma Q_{1}^\dagger
\end{equation}
where
\begin{equation}
\begin{split}
&W_{1}=
\left(
\begin{array}{ccc}
-0.958 & 0.194 & -0.21 \\
-0.204 & -0.979 & 0.0279 \\
-0.200 & 0.0696 & 0.977
\end{array}
\right) \\
&\Sigma=
\left(
\begin{array}{ccc}
1 & 0 & 0 \\
0 & 0.991 & 0 \\
0 & 0 & 0.991
\end{array}
\right) \\
&Q_{1}=
\left(
\begin{array}{ccc}
-0.754 & -0.504 & -0.422 \\
0.646 & -0.452 & -0.615 \\
0.119 & -0.736 & 0.666 
\end{array}
\right).
\end{split}
\end{equation}
We will parametrize only the most interesting case of unitary dilation of a minimal dimensionality, and hence of a minimal number of additional neutrinos, i.e., the number of singular values strictly less than 1. Since the matrix $\Sigma$ determines the singular values of $V'$, this number equals 2. Hence it is possible to construct unitary dilation $U$ of the minimal dimension $5 \times 5$.

To complete the construction, we are left only with two free unitary $2\times 2$ ``parameters'' $W_{2}$ and $Q_{2}$, and for the sake of this example we choose them randomly:
\begin{equation}
\begin{split}
&W_{2}=
\left(
\begin{array}{cc}
-0.619 & 0.785 \\
0.785 & 0.619 
\end{array}
\right), \\
&Q_{2}=
\left(
\begin{array}{cc}
0.250 & -0.968 \\
0.968 & 0.250
\end{array}
\right).
\end{split}
\end{equation} 
Having all ingredients and making all necessary calculations, we find the following form of the unitary dilation of $V'$
given by $U= (W_1\oplus W_2) \Sigma (Q_1\oplus Q_2)^\dagger$:

\begin{equation}
\begin{small}
U=
\left(
\begin{array}{ccc|cc}
0.822 & 0.549 & 0.147 & 0.0207 & 0.0322 \\
-0.469 & 0.521 & 0.701 & 0.0292 & -0.128 \\
0.311 & -0.643 & 0.687 & -0.129 & -0.0237 \\ \hline
-0.041 & -0.0399 & 0.121 & 0.599 & 0.788 \\
0.0788 & -0.109 & -0.009 & 0.788 & -0.599
\end{array}
\right).
\end{small}
\end{equation}
Since we have freedom of choice of two unitary matrices, it is necessary to check how this choice influences the result. Let us generate randomly another pair:
\begin{equation}
\begin{split}
&W_{2}'=
\left(
\begin{array}{cc}
0.346 & -0.938 \\
0.938 & 0.346 
\end{array}
\right), \\
&Q_{2}'=
\left(
\begin{array}{cc}
-0.888 & -0.461 \\
-0.461 & 0.888
\end{array}
\right).
\end{split}
\end{equation}
Then we get the following unitary dilation:
\begin{equation}
\begin{small}
U'=
\left(
\begin{array}{ccc|cc}
0.822 & 0.549 & 0.147 & 0.0101 & 0.0369 \\
-0.469 & 0.521 & 0.701 & -0.115 & -0.0638 \\
0.311 & -0.643 & 0.687 & 0.0686 & -0.112 \\ \hline
0.0149 & 0.0716 & -0.112 & 0.124 & -0.984 \\
0.0867 & -0.091 & -0.0464 & -0.984 & -0.124
\end{array}
\right)\ .
\end{small}
\end{equation}
The matrices $U,U'$ differ by both the off-diagonal block and by the bottom-diagonal block. However, the scale of the off-diagonal block is comparable in both cases. The reason for this lies in the fact that to construct each of these blocks we use $C$ and $S$ fixed by the singular value decomposition of $V'$ matrices. The biggest difference can be observed in the bottom diagonal block since only the matrix $S$ is fixed in both cases. However, the global scale of each block (global in a sense of the Frobenius norm, which is an entrywise norm defined in \cite{books/daglib/0019187}) is conserved in each of these cases. Since this norm is unitarily invariant, the choice of $W_{1,2}$ and $Q_{1,2}$ does not change its value.

The dilation procedure described above is based exclusively on mixing matrices. In contrast, there are constructions in the literature which refer in addition to the mass spectrum; see for instance,  \cite{Casas:2001sr,Agostinho:2017wfs,Fong:2017gke}. 
  In the approach taken in this work, the information on the number of additional neutrinos, i.e., the dimension of the complete unitary mixing matrix, is nicely seen through the number of nonunit singular values. 
As discussed in the Setting section, our approach based on singular values and the dilation procedure is general, no matter if extra neutrino states are heavy (e.g., see-saw mechanisms) or light and sterile. 
As far as the present situation in neutrino physics is concerned, the minimal 3+1 neutrino scenario is still not excluded, though LSND and MiniBoone results make it a less probable scenario. For a global analysis see \cite{Donini:2011jh, Gariazzo:2017fdh, Dentler:2018sju}. Here we considered an example of an extension to $3+2$ dimensions. However, it is still possible to find one of the singular values strictly less than 1, while the remaining two are equal to 1, so extensions to the $3+1$ model are possible in our rough estimations. For complete future studies of dilations, among others, the following issues can be addressed in more detail. First, our $3+2$ example is only one of many elements of the complete convex hull $\Omega$. It would be interesting to map out if there are well-defined regions in the interval matrix with extensions to 4 dimensions, while others have a minimal dimension equal to 5 dimensions.
Second, in complete studies, CP-breaking mixings should be included. Finally, error estimation is crucial for future, more refined studies. So far we rely on the {Weyl} estimation \eqref{weyl}, which is a rough estimate. Our basic description indicates possible directions for further studies through the notion of singular values.  
 
\begin{figure}[t!]
\begin{center}
\includegraphics[scale=0.6]{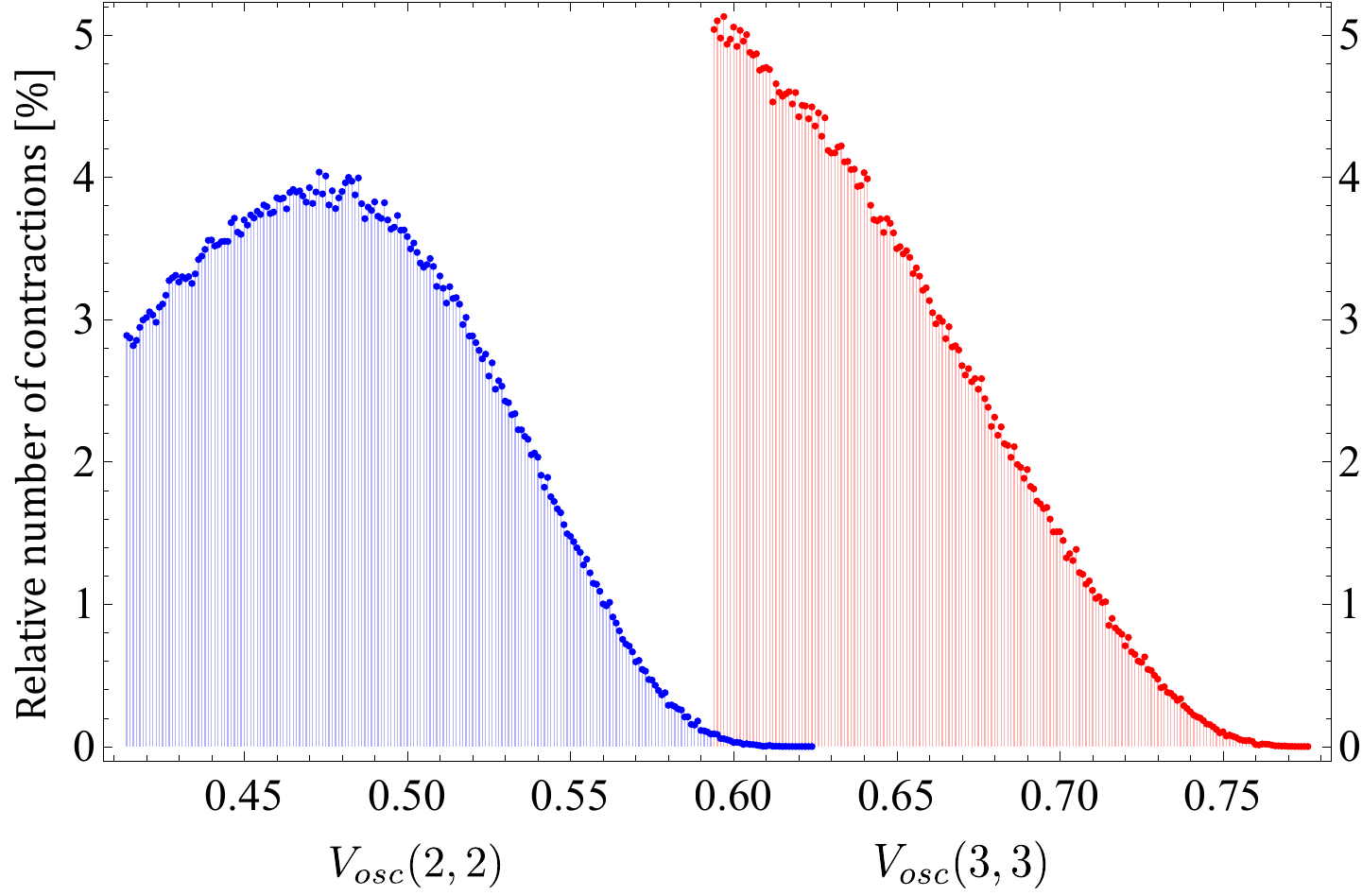} 
\caption{Typical distributions of contractions for the (2,2) and (3,3) elements of the $\uosc$ peaked inside and at the edge of allowed ranges. 
}
\label{dist}
\end{center}
\end{figure}

\section{Neutrino data analysis} 
The interval mixing matrix \eqref{ranges} contains unitary (SM) and nonunitary contraction (BSM) and unphysical matrices (the latter have to be discarded by the contraction property).
We have found that statistically about 4\% of matrices $V\in \uosc$ are contractions, while some unphysical ones have norms as large as $\|V\|=1.178$.
This result was obtained within $0.003$ accuracy, by uniformly sampling elements of the intervals of $\uosc$ with sufficiently high statistics.
All calculations presented in this work has been made in \texttt{Mathematica} \cite{Mathematica}.
The statistical analysis of distributions of contractions in $\uosc$ has been performed under the assumption of mixing parameter errors having a uniform distribution. 
This implies that values in $\uosc$ were also treated uniformly. A discretization of intervals in $\uosc$ was made with a step $0.001$ to match the precision of extreme values. Up to $10^9$ randomly generated matrices have been produced within $\uosc$ ranges for which singular values have been found. Next, the largest singular value for each random matrix was compared to the number $1.003$, to be consistent with the precision ensured by the stability of singular values,  
splitting in this way matrices into two sets of contractions and noncontractions.  

Likewise,  we analyze  distributions of contractions for a given element in $\uosc$. We  fix a value of one of the elements of $\uosc$ and then randomly generate matrices and make the same analysis as above. 
As an illustration, Fig.~\ref{dist} presents contraction distributions for two exemplary matrix elements taken from the full interval matrix. While one may argue that these diagrams show only statistical density, Proposition~\ref{poly} in the Appendix~C shows that there is a sharp matrix boundary (surface in $\mathbb{C}^9$) with an interior composed solely of contractions. 

If we shrink errors in \eqref{exp3s} to 1$\sigma$ C.L., we get $11\%$ of contractions, instead of 4\% discussed at the beginning of this section. 
Narrowing the angle ranges \eqref{exp3s} usually increases the amount of contractions in $\uosc$; however, for arbitrary angle ranges this does not always occur.
Concerning new physics, it is interesting to ask how strong a contraction can  be found  in $\uosc$.
The minimal value of the norm for $V\in \uosc$ is $\|V_{\rm min}
\|=0.961$ and can be obtained by sufficiently fine discretization of $\uosc$. 
Alternatively, 
$\|V_{\rm min}\|$ can readily be obtained by semidefinite programming, which is a very useful numerical tool when analyzing properties of interval matrices \cite{boyd2004convex}.

It should be stressed that we apply our methods to data in order to illustrate our matrix machinery in applications to neutrino mixing matrices but do not attempt to make a definite analysis. 
We have made rough estimations based on a construction where experimental data and PMNS formalism are used, though as mentioned already, the interval matrix can be obtained even directly without restriction to PMNS parametrizations when nonunitarity is assumed from the very beginning 
\cite{Parke:2015goa}. This interesting and universal option is left for separate and detailed future studies.

\section{Quark data analysis}
Our scheme in Fig.~\ref{scheme} is general enough to be used in the quark sector as well.
For quarks the unitary CKM mixing matrix \cite{Cabibbo:1963yz,Kobayashi:1973fv} can be parametrized in the same way as the PMNS mixing matrix for neutrinos: 
\begin{equation}  
\begin{split}
&V_{CKM} = 
\left(
\begin{array}{ccc}
V_{ud} & V_{us} & V_{ub} \\
V_{cd} & V_{cs} & V_{cb} \\
V_{td} & V_{ts} & V_{tb} 
\end{array}
\right)= 
\left(
\begin{array}{ccc}
    1 & 0 & 0 \\
    0 & c_{23}  & {s_{23}} \\
    0 & -s_{23} & {c_{23}}
\end{array}
\right) \\
&\times 
\left(
\begin{array}{ccc}
    c_{13} & 0 & s_{13} e^{-i\delta_\text{}} \\
    0 & 1 & 0 \\
    -s_{13} e^{i\delta_\text{}} & 0 & c_{13}
\end{array}
\right)
\left(
\begin{array}{ccc}
    c_{12} & s_{12} & 0 \\
    -s_{12} & c_{12} & 0 \\
    0 & 0 & 1
\end{array}
\right).
\end{split}
\end{equation}
Experimental results have established the following hierarchy of mixing parameters:
\begin{equation}\label{eq:53}
s_{13} \ll s_{23} \ll s_{12} \ll 1.
\end{equation}
Due to this order it is convenient to present $V_{CKM}$ in an approximate parametrization proposed by Wolfenstein \cite{Wolfenstein:1983yz}, which reflects the above hierarchy. The mixing parameters (\ref{eq:53}) are connected with Wolfenstein parameters in the following way:
\begin{equation}
\begin{split}
&s_{12}= \lambda, \\
&s_{23}= A\lambda^{2}, \\
&s_{13}e^{i\delta}=A\lambda^{3}(\rho+i\eta).
\end{split}
\end{equation}
This results in the following structure of the quark mixing matrix:
\begin{eqnarray}\label{eq:55}
V_{CKM}&=&
\left(
\begin{array}{ccc}
1-\frac{\lambda^{2}}{2} & \lambda & A\lambda^{3}(\rho -i \eta) \\
-\lambda & 1-\frac{\lambda^{2}}{2} & A\lambda^{2} \\
A\lambda^{3}(1-\rho -i \eta) & -A\lambda^{2} & 1
\end{array}
\right) \nonumber \\&+& \mathcal{O}(\lambda^{4}).
\end{eqnarray}
We are interested in how contractions are distributed within $V_{CKM}$ with respect to experimental values of the mixing parameters \cite{Olive:2016xmw}:
\begin{equation}
\begin{split}
\lambda=0.22506 \pm 0.00050, \\
A=0.811 \pm 0.026, \\
\bar{\rho}=0.124^{+0.019}_{-0.018}, \\
\bar{\eta}=0.356 \pm 0.011,
\end{split}
\end{equation}
where $\bar{\rho}=\rho(1-\lambda^{2}/2)$ and $\bar{\eta}=\eta(1-\lambda^{2}/2)$. The application of the above results to (\ref{eq:55}) gives us the following experimental intervals for elements of the mixing matrix:
\begin{equation}
\begin{split}
&V_{ud} \in \left[ 0.97456 , \ 0.97478 \right] \\
&V_{us} \in \left[ 0.22456 , \ 0.22556 \right] \\
&V_{ub}\in\left[0.00097 - 0.00362i , \ 0.00141 -0.00315i\right] \\
&V_{cd} \in \left[ -0.22556 , \ -0.22456 \right] \\
&V_{cs} \in \left[ 0.97456 , \ 0.97478 \right] \\
&V_{cb} \in \left[ 0.0396 , \ 0.0426 \right] \\
&V_{td}\in\left[0.00758-0.0362i , \ 0.00856-0.00315i\right] \\
&V_{ts} \in \left[ -0.0426 , \ -0.0396 \right] \\
&V_{tb}=1 
\end{split}
\label{ckmint}
\end{equation}
where intervals in the case of $V_{ub}$ and $V_{td}$ are to be understood as complex rectangles.
Our statistical analysis reveals that all matrices within $V_{CKM}$ are contractions with $0.002$ accuracy. Analysis of values of operator norm gives the following statistical result:
\begin{equation}
\begin{split}
&6 \% \ \rm{of} \ \Vert V_{CKM} \Vert = 1.002, \\
&94 \% \ \rm{of} \ \Vert V_{CKM} \Vert = 1.001.
\end{split}
\end{equation}
Let us recall that in the neutrino case, minimal and maximal deviations from unity are 0.961 and 1.178, respectively. 
It shows how much, as far as the precision of the analysis in the neutrino sector is concerned, still must be done there.

It is interesting that a vanishing fraction of matrices within $V_{CKM}$ has a norm strictly less than one. This can be a sign that the only contractions in the quark sector are unitary matrices. However, since we have used only the leading order of the mixing matrix expressed by the Wolfenstein parameters, additionally  more refined analysis of this sector is necessary.  
In principle, we do not have to rely on the Wolfenstein parametrization and the analysis can be done directly on quark data in the form of an interval matrix.  
At the LHC there are already direct measurements of $V_{tq}$ ($q=d,s,b$) by studying top production as well as its decays and charge asymmetry \cite{Khachatryan:2014nda,Sirunyan:2016cdg,Aaboud:2017pdi}. 
Our approach based on the interval matrix will become very interesting in the context of future collider experiments, like FCC-hh, with center of mass energies a few times larger than those of the LHC \cite{Golling:2016gvc}, where all elements of the interval mixing matrix can be probed directly with much better precision.

\section{Summary and Outlook}

We have shown how to recover physically admissible mixing matrices from the interval matrix representation of neutrino or quark data, 
namely, any contraction matrix within the interval matrix is physical and has properly correlated matrix elements.
This characterization is complete, as any contraction can be completed into a unitary matrix via a unitary dilation procedure which yields an extension of minimal dimension.  
The approach is universal in the sense that it does not invoke any specific  parametrization and is based on general features of the interval matrices.
Physical mixing matrices consistent with the experiment are shown to have the structure of a convex hull over admissible PMNS matrices.

Singular values play a special role in our analysis. The general observation is that whenever we find singular values smaller than one, it is a \emph{signature of BSM}.
This observable seems to be an interesting alternative to other quantifiers of unitarity-breaking so far employed in literature.
We are commenting on possible analysis in the quark sector and our estimations based on Wolfenstein parametrization point out very little space for nonunitary effects there.

Finally, assuming a BSM scenario, we show how to construct a unitary mixing matrix of minimal dimension larger than 3 consistent with data. It allows us  in particular to construct a dilation procedure to determine the \emph{minimal} number of extra neutrino species, compatible with experimental data in a BSM scenario. 
This is potentially a very fertile area of study. 
Should a BSM signal be found, dilation theory will be a promising point of departure for further analysis.
Of course our studies are not complete with this commencing paper. The estimation of errors to judge unambiguously deviations of singular values from unity will be crucial in the future. In this work we estimate errors on singular values through Weyl inequalities.

Our methods are based on advanced matrix analysis, studying the singular values of mixing matrices.
We apply a model-independent analysis based on the interval matrix to the present data, in a way that may become significant in future experiments that will measure entries of this matrix directly. It can also be useful through Propositions 1 and 2 to cross-check with other analyses based on specific parametrizations, since the contraction condition is easy to apply.

We shall go further in this direction and merge our studies on mixings (eigenfuncion problems) with masses (eigenvalues). 
For instance we could study the angle between subspaces of the mass matrices to connect neutrino masses with their mixing. This approach is closely related to the methodology presented in our work. 
Moreover, a separate analysis of the properties of the neutrino mass spectrum could be done exclusively. For this we might adopt many advanced methods of matrix analysis, e.g., Gershgorin circles. 
Clearly further potential for practical applications of our procedures is there.

\section*{Acknowledgements} 
This publication was made possible through the support of a grant from the John Templeton Foundation (Grant No. 60671) and the support of the Polish National Science Centre (NCN), Grant  No.~DEC-2013/11/B/ST2/04023.  We would like to thank B. Kayser, C. Giunti,  M. Laveder, T. Riemann and M.
Zra\l ek for useful comments.

\section*{Appendix}
  
In the following we give the technical details supplementing the results of the main text. 
We begin in Appendix A by providing more details on contractions as principal submatrices of unitary matrices.
Then in Appendix B we give a very simple example of how contractions allow us to restrict parametrizations of mixing matrices.
We then provide a section in Appendix C describing the interval matrices within convex geometry.
  In Appendix D we provide a description of the theory of matrix dilations. 
In Appendix E various nonunitary parametrizations are classified. Their relation to contractions is discussed.

\section*{A. Contractions}
A matrix norm is a function $\Vert \cdot \Vert$ from the set of all complex matrices  into $\mathbb{R}$ that satisfies for any $A,B\in M_{n\times n}$ the following properties:
\begin{equation}
\label{norm}
\begin{split}
& \Vert A \Vert \geq 0 \  \begin{rm}and\end{rm} \  \Vert A \Vert =0 \Leftrightarrow A=0, \\
& \Vert \alpha A \Vert= \vert \alpha \vert \Vert A \Vert, \alpha \in \mathbb C, \\
&\Vert A + B \Vert \leq \Vert A \Vert + \Vert B \Vert, \\
&\Vert A  B \Vert \leq \Vert A \Vert  \Vert B \Vert. \\
\end{split}
\end{equation}
In other words, a matrix norm is a vector norm (first three conditions in (\ref{norm})) with an additional condition of submultiplicativity.
The most important norm in our work is the operator norm
 $\Vert A \Vert  = \max_{\Vert x \Vert = 1} \Vert Ax \Vert$,
 for which one can prove that it is equal to the largest singular value $\Vert A\Vert=\max_i \sigma_{i}(A)$, where we have $\sigma_i(A)=\sqrt{\lambda_i(AA^\dagger)}$; i.e., singular values are the positive square roots of the eigenvalues of $AA^\dagger$ denoted by $\lambda_i(AA^\dagger)$.
We note that there exist other matrix norms that bring different properties into focus \cite{books/daglib/0019187} but are less important for mixing matrices.

We now consider any principal submatrix $V$ of a unitary matrix $U$ and show that it is a contraction, i.e.,  $\|V\|\le1$ in the operator norm.

\begin{prop} 
         \label{prop:sub} 
          If $A\in M_{n\times n}$ and $B\in M_{m\times m}$ is any principal submatrix of $A$, then
          \begin{equation}
           \Vert B\Vert\leq\Vert A\Vert.
          \end{equation}
         \end{prop}
         Proof. It is straightforward to see that for any unit $x\in\mathbb{C}^m$ there is a unit embedding $y\in\mathbb{C}^n$ of $x$ such that
         \begin{equation}
          \Vert Bx\Vert=\Vert Ay\Vert
         \end{equation}
         (namely, by inserting zeros at entries of $y$ corresponding to columns of $A$ deleted to obtain $B$).
         Furthermore, the range of this embedding is a subspace of $\mathbb{C}^n$, and hence 
         \begin{equation}
          \sup_{\Vert x\Vert=1}\Vert Bx\Vert\leq \sup_{\Vert y\Vert=1}\Vert Ay\Vert,
         \end{equation}
         which gives the result.\hfill$\blacksquare$
         
         The next observation is almost trivial, yet is crucial in the analysis of neutrino mixing matrices in the main text.
         \begin{cor}
          Let $U\in M_{n}$ be unitary. Then $\Vert U\Vert=1$ and any submatrix $V$ of $U$ is a contraction.        \end{cor}
Proof. The equality $\lambda_i(UU^\dagger)=\lambda_i(I)$ implies that $\|U\|=1$. By Proposition \ref{prop:sub}, for any submatrix $V$ of $U$ it holds that $\Vert V\Vert\leq\Vert U\Vert= 1$; hence $V$ is a contraction.\hfill$\blacksquare$

  \section*{B. Unitarity and contractions: Toy example}
Here we provide more details on problems occurring when studying nonunitary $\upmnsca$ through a particular parametrization.
For $\upmns$ it holds that the sum of probability of neutrino oscillations equals 1:
\begin{equation}
\sum\limits_{\alpha} P_{i \alpha}=1,\;\;\; e.g. \;\;\;P_{ee}+P_{e \mu}+P_{e \tau}=1.
\end{equation}

However, for a nonunitary $U$ analogous relation is not 
fulfilled. Let us see it in a simple case of two flavors (the same can be done for a dimension-3 modified $\upmns$ matrix),
when $U$ is defined as ($\Theta_2=\Theta_1+\epsilon$)

\begin{equation}
U=
\begin{pmatrix} \cos{\Theta_1} &  \sin{\Theta_1} \\
-\sin{\Theta_2} &  \cos{\Theta_2}  
\end{pmatrix}
.\label{utoy}
\end{equation}

In this case we get $\Delta_{ij} \propto  (m^2_{i}-m^2_{j}) \frac{L}{E}$ and
\begin{eqnarray}
\sum\limits_{\alpha=e,\mu} P_{e \alpha} & = &
P_{ee}+P_{e \mu}\\ &=& 1+4 \epsilon \sin^2{\Delta_{21}}  
\sin{\Theta_1} \cos{\Theta_1} \cos{2 \Theta_1} +
{\cal{O}} (\epsilon^2), \nonumber \\&& \nonumber
\end{eqnarray}
\begin{eqnarray}
\sum\limits_{\alpha=e,\mu} P_{\mu \alpha} & = &
P_{\mu e}+P_{\mu \mu}  \\&=&
1-4 \epsilon \sin^2{\Delta_{21}}  
\sin{\Theta_1} \cos{\Theta_1} \cos{2 \Theta_1}
+{\cal{O}} (\epsilon^2). \nonumber
\end{eqnarray}

We can see that the sum  can be either  larger or smaller than 
1. This example was given in \cite{Czakon:2001em}; however, no clue at that time was given about how to interpret possible results when the sum of probabilities does not equal 1. Here we show that matrix (\ref{utoy}) is not the right way to parametrize BSM effects. Let us find the norm which helps us to interpret the matrix (\ref{utoy}). 

First, we calculate $UU^{T}$ and $U^{T}U$ for (\ref{utoy}), $s(c)_a \equiv \sin (\cos)\Theta_{a} $, as 

\begin{equation}
UU^{T}=
\left(
\begin{array}{cc}
1 & s_{1}c_{2} -s_{2}c_{1} \\
s_{1}c_{2} -s_{2}c_{1} & 1
\end{array}
\right),
\end{equation}

\begin{equation}
U^{T}U=
\left(
\begin{array}{cc}
c_{1}^{2}+ s_{2}^{2} & c_{1}s_{1}-s_{2}c_{2} \\
c_{1}s_{1}-s_{2}c_{2} & s_{1}^{2}+c_{2}^{2}
\end{array}
\right).
\end{equation}
As for the real $A$, we have $\Vert A^{T}A \Vert = \Vert AA^{T} \Vert = \Vert A \Vert^{2}$; we can focus only on one of these products. We write  $UU^{T}$   in the following form:
\begin{equation}
\begin{split}
UU^{T}&=
\left(
\begin{array}{cc}
1 & s_{1}c_{2} -s_{2}c_{1} \\
s_{1}c_{2} -s_{2}c_{1} & 1
\end{array}
\right) \\
&=
\left(
\begin{array}{cc}
1 & 0 \\
0 & 1
\end{array}
\right)+
\left(
\begin{array}{cc}
0 & s_{1}c_{2} -s_{2}c_{1} \\
s_{1}c_{2} -s_{2}c_{1} & 0
\end{array}
\right).
\end{split}
\end{equation}
This can be simplified into  
\begin{equation}
UU^{T}=
\left(
\begin{array}{cc}
1 & 0 \\
0 & 1
\end{array}
\right)+
\left(
\begin{array}{cc}
0 & s_{3} \\
s_{3} & 0
\end{array}
\right) \equiv
I + B
\end{equation}
where $s_{3} \equiv \sin \Theta_{3} = \sin ( \Theta_{1} -\Theta_{2})$. $B$ is symmetric and its eigenvalues are equal to $\pm s_{3}$. Let $V$ be a unitary matrix such that $V^{T}BV=D = \mathrm{diag}(s_{3}, -s_{3})$. Since the operator norm is unitarily invariant \cite{books/daglib/0019187}, we write
\begin{eqnarray}
\Vert UU^{T} \Vert &=& \Vert I + B \Vert = \Vert V^{T} (I + B) V \Vert = \Vert I + V^{T}BV \Vert \nonumber \\&=& \Vert I + D \Vert .
\end{eqnarray}
Since $I + D$ equals
\begin{equation}
\left(
\begin{array}{cc}
1+s_{3} & 0 \\
0 & 1-s_{3}
\end{array}
\right),
\end{equation}
its operator norm, i.e., the largest singular value, is equal to
\begin{equation}
\begin{split}
1+s_{3} \quad if \ s_{3} \geq 0, \\
1-s_{3} \quad if \ s_{3} < 0.
\end{split}
\end{equation}
So we can see that by adding $B$ to the identity matrix we cannot decrease the operator norm: 
\begin{equation}
1= \Vert I \Vert \leq \Vert I + B \Vert = \Vert UU^{T} \Vert=1+|s_3|.
\end{equation}
Thus
\begin{equation}
\Vert U \Vert \geq 1.
\label{contrex}
\end{equation} \\
As discussed in the main text, a physically meaningful theory should include only fields for which contraction relation $\Vert U \Vert \leq 1$ is fulfilled, and $\Vert U \Vert > 1$, being a part of some more complex complete theory based on unitarity  cannot describe BSM effects at all. 
The result (\ref{contrex}) implies that not all parametrizations which violate unitarity are a proper choice, and a toy mixing matrix (\ref{utoy}) is superfluous from the physical point of view. It fulfills $\Vert U \Vert =1$  for $ \epsilon=0 $, but then a trivially unitary matrix is recovered.

\section*{C. Convex geometry and interval matrix analysis}

Here we gather necessary facts and comments that refer to convex geometry, which plays a crucial role in the paper in a twofold way: it gives a very convenient parametrization of contraction matrices (see Theorem \ref{th:k-m}) and provides some decisive conditions on distributions of (non)contractions in interval matrices (see Proposition \ref{poly}).
\\
\begin{defi} \cite{magaril2003convex}
A nonempty set $A \subset \mathbb{R}^{n}$ is convex if, along with any of its two points $x$ and $y$, it contains the line segment $[x,y]$, i.e., the set
\begin{equation}
[x,y]= \lbrace z \in \mathbb{R}^{n}: z=\alpha x+ (1-\alpha)y, 0\leq \alpha \leq 1 \rbrace .
\end{equation}
\end{defi}

\begin{defi} \cite{Krantz2014}
Let $A \subseteq \mathbb{R}^{n}$ be any set. The convex hull of $A$ denoted by $\mathrm{conv}(A)$ is the intersection of all geometrically convex sets that contain $A$.
\end{defi}

\begin{lem} \cite{Matouek:2006:UUL:1214763}
The convex hull of the set $A \subseteq \mathbb{R}^{n}$ equals the set
\begin{equation}
\begin{split}
\mathrm{conv}(A)= \lbrace \sum_{i}^{m} &\alpha_{i}x_{i} \mid m \geq 1, x_{1},...,x_{m} \in A \subseteq \mathbb{R}^{n}, \\
 &\alpha_{1},...,\alpha_{m} \geq 0, \sum_{i}^{m}\alpha_{i}=1 \rbrace.
\end{split}
\end{equation}
of all convex combinations of finitely many points of $A$.
\end{lem}
The following theorem states that there is an analogue of a linear span in convex geometry, such that the span is over all extreme points of the set $A$, i.e., points that are not interior points of any line segment lying entirely in $A$.
\begin{theorem}\label{th:k-m}{(Krein-Milman)} \cite{Krantz2014} \\
Let $X$ be a topological vector space in which the dual space $X^{*}$ separates points. If $A$ is a compact, convex set in $X$, then $A$ is a closed, convex hull of its extreme points.
\label{th:krein}
\end{theorem}

\begin{prop}
\label{poly}
 Once a set of matrix contractions is given, the convex hull with vertices at this set contains only contractions. 
\end{prop}

Proof. Let $n$ be fixed and consider the nonempty polytope $P=\times_{k=1}^{n^2}[a_k,b_k]$.
To every $p\in P$ we associate a matrix $A^{(p)}$ with entries $A^{(p)}_{i,j}=p_{\zeta(i,j)}$, where  $\zeta: [n]^2\rightarrow [n^2]$ is the bijective map defined by $\zeta(j,k)= (j-1)n+k$. 
We will show that the subset of matrices based in $P$ is convex; i.e., for $p,q\in P$, if $\|A^{(p)}\|\le1$ and $\|A^{(q)}\|\le1$, then $\|A^{(\lambda p+(1-\lambda)q)}\|\le1$ for $0\le\lambda\le1$.

We now explicitly calculate 
\begin{align}
A^{(\lambda p+(1-\lambda)q)}_{i,j}&=\lambda p_{\zeta(i,j)}+(1-\lambda)q_{\zeta(i,j)}\\&=\lambda A^{(p)}_{i,j}+(1-\lambda)A^{(q)}_{i,j}
\end{align}
From the triangle inequality we obtain
\begin{align}
\|A^{(\lambda p+(1-\lambda)q)}\|&\le
\lambda \|A^{(p)}\|+(1-\lambda)\|A^{(q)}\|\le 1\;.
\end{align}
This means that if one verifies that for a set of points $p_1,\ldots,p_N$ the matrices are contractions, then for all matrices in the convex hull $p\in \mathrm{conv}\{p_1,\ldots,p_N\}$ the matrix $A^{(p)}$ will be a contraction.\hfill$\blacksquare$

In particular this means that if one checks that contractions are vertices of some $Q=\times_{k=1}^{n^2}[a'_k,b'_k]\subseteq P$, then no matrix inside $Q$ will have a norm larger than $1$. 

\section*{D. Unitary dilations}
To find a complete theory for BSM mixing matrices we need to find a matrix that has a nonunitary $V$ as a principal submatrix and is unitary.
In 1950 Halmos \cite{Halmos:1950} noticed that any contraction $A$ acting on a Hilbert space $H$ can be dilated to a unitary operator which acts on $H \oplus H$ space by
\begin{equation}
\label{eq:halmos}
U=
\left(
\begin{array}{cc}
A & (I-AA^{\dag})^{1/2} \\
(I-A^{\dag}A)^{1/2} & -A^{\dag} 
\end{array}
\right).
\end{equation} 
A few years later, Sz.- Nagy \cite{Sz.-Nagy:1953} generalized this idea.
In the Halmos construction we see that for an $n \times n$ matrix $A$ its unitary dilation $U$ will have dimension $2n \times 2n$. 
There exists a further theorem \cite{Thomson:1982} which allows us to dilate a contraction to a unitary matrix of possibly lower dimension than $2n$, yet some additional conditions must be satisfied.
\begin{theorem}[\cite{Thomson:1982}]
\label{th:min}
 A matrix $A\in M_{k\times k}$ is a principal submatrix of a unitary $U\in M_{n\times n}$ iff A is a contraction and $m=\mathrm{rank}(I-A^{\dagger}A)\leq\mathrm{min}\{k,n-k\}$. 
\end{theorem}
Recall that the rank of a matrix can be defined as the number of its nonzero singular values.
We use this theorem to show that $m$ is optimal.
\begin{cor}
\label{cor:min}
 Let $A,U$ be as above and $m=\mathrm{rank}(I-A^{\dagger}A)$. Then the minimal dimension of $U$ is $n=k+m$.
\end{cor}
Proof. Suppose $n<k+m$. From Theorem \ref{th:min} we have $m\leq\mathrm{min}\{k,n-k\}$, and hence $m\leq n-k$ in particular. Thus $n\geq k+m$, which contradicts the assumption.
 \hfill$\blacksquare$
 
In the main text we construct the minimal extension and use the fact that $\mathrm{rank}(I-A^{\dagger}A)$ is equal to the number of singular values of $A$ strictly less than one, which is a direct consequence of a rank definition given above.
The construction is achieved through the CS decomposition of unitary matrices.
In \cite{Allen06matrixdilations} it has been shown how the Halmos construction (\ref{eq:halmos}) is a particular example of the CS decomposition. This construction in its generality allows for dilations of dimension determined by Corollary \ref{cor:min}. Again, singular values play a crucial role here.
 \\
\begin{theorem}[\cite{Allen06matrixdilations}]
\label{th:2}
Let the unitary matrix $U \in M_{(n+m) \times (n+m)}$ be partitioned as
\begin{equation}
U=
\begin{blockarray}{ccc}
n & m & \\ 
\begin{block}{(cc)c}
U_{11} & U_{12} & \ n \\
U_{21} & U_{22} & \ m \\
\end{block}
\end{blockarray}
\end{equation}
If $m \geq n$, then there are unitary matrices $W_{1}, Q_{1} \in M_{n \times n}$ and unitary matrices $W_{2}, Q_{2} \in M_{m \times m}$ such that
\begin{equation}
\begin{split}
&\left(
\begin{array}{cc}
U_{11} & U_{12} \\
U_{21} & U_{22}
\end{array}
\right)= \\
&\left(
\begin{array}{cc}
W_{1} & 0 \\
0 & W_{2}
\end{array}
\right)
\left(
\begin{array}{c|cc}
C & -S & 0 \\ \hline
S & C & 0 \\
0 & 0 & I_{m-n}
\end{array}
\right)
\left(
\begin{array}{cc}
Q_{1}^{\dag} & 0 \\
0 & Q_{2}^{\dag}
\end{array}
\right),
\end{split}
\end{equation}
where $C \geq 0$ and $S \geq 0$ are diagonal matrices satisfying $C^{2} + S^{2}=I_{n}$.
\end{theorem}
If $n \geq m$, then it is possible to parametrize a unitary dilation of the smallest size.
\begin{cor}\label{cor:2}
The parametrization of the unitary dilation of smallest size is given by
\begin{equation}
\begin{split}
&
\left(
\begin{array}{cc}
U_{11} & U_{12} \\
U_{21} & U_{22}\
\end{array}
\right)=\\
&
\left(
\begin{array}{cc}
W_{1} & 0 \\
0 & W_{2} 
\end{array}
\right)
\left(
\begin{array}{cc|c}
I_{r} & 0 & 0 \\
0 & C & -S \\ \hline
0 & S & C
\end{array}
\right)
\left(
\begin{array}{cc}
Q_{1}^{\dag} & 0 \\
0 & Q_{2}^{\dag}
\end{array}
\right),
\end{split}
\end{equation}
where $r=n-m$ is the number of singular values equal to 1 and $C=diag(\cos \theta_{1},...,\cos \theta_{m})$, with $\vert \cos \theta_{i} \vert <1$ for $i=1,...,m$.
\end{cor}

\section{E. BSM parametrizations of neutrino mixings and contractions}

There exist three different matrix factorizations that decompose a matrix into a product of two matrices of which one is unitary, namely \cite{books/daglib/0019187, 2017arXiv170400122G},
\begin{enumerate}
\item Polar decomposition,
\item QR decomposition,
\item Mostow decomposition.
\end{enumerate}
The first two are used frequently in neutrino physics in the context of parametrization of nonunitarity effects in the neutrino mixing matrix. These are the polar decomposition and a modified version of the QR decomposition. Thus, let us take a closer look at these two parametrizations. The polar decomposition factorizes a given square matrix $A$ into the following product,
\begin{equation}
A=PU,
\end{equation}
where matrix $P$ is a positive semidefinite Hermitian matrix and $U$ is a unitary matrix. The polar factor $P$ is uniquely determined and is given by $\sqrt{AA^{\dag}}$, while the unitary part is also uniquely determined if the initial matrix is nonsingular.

To our knowledge an application of the polar decomposition to parametrize a deviation from unitarity in the neutrino sector appears for the first time in \cite{Antusch:2006vwa}. There, the polar factor is further decomposed in the following way, 
\begin{equation}\label{eta1}
P=I-\eta,
\end{equation}
where a matrix $\eta$ describes the deviation from unitarity of the neutrino mixing matrix. As we recall from the main text, physical mixing matrices must be contractions, i.e., matrices with spectral norm less than or equal to one or equivalently with the largest singular value less than or equal to one. Let us notice that in general the polar decomposition does not provide this property. To see this, let us look at a simple example, where we take the matrix $\eta$ in a simple diagonal form \footnote{Since by the unitarily invariance of the norm the unitary part is irrelevant here, we can focus only on the polar factor given by \eqref{eta1}.}
\begin{equation}
\eta=
\left(
\begin{array}{cc}
\epsilon & 0 \\
0 & -\epsilon
\end{array}
\right),
\end{equation}
where $0<\epsilon \leq 1$.

Observe that this results in a positive semidefinite matrix $P=I-\eta$, which is necessary for a polar factor. However, such $P$ is not a contraction since one singular value will be always larger than one, independently of how small $\epsilon$ is.

Recently, the polar factor in the form of \eqref{eta1} was identified with a matrix $I-\frac{\Theta \Theta^{\dag}}{2}$, which arises in the context of the complete unitary mixing matrix \cite{Antusch:2009pm} (for a similar construction see also \cite{Blennow:2011vn, Agostinho:2017wfs}). Thus in a scenario such that the complete unitary mixing matrix is considered, the polar factor is by definition a contraction. In this approach, to ensure that the polar factor $I-\eta$ is a contraction, a necessary condition for the matrix $\eta$ follows in a form of a positive semidefinite matrix. Using the fact that the operator norm is unitarily invariant, it can be shown that for sufficiently small entries of the matrix $\eta$ also the inverse is true; i.e., if the matrix $\eta$ is positive semidefinite, then $P=I-\eta$ must be a contraction. Scenarios that employ such unitarity-breaking constructions are usually called top-down approaches.

The second of the currently used factorizations in neutrino physics is the QR decomposition. It factorizes a given matrix into a product of a unitary matrix $Q$ and an upper triangular matrix $R$ and was proposed as a parametrization of the neutrino mixing matrix in \cite{Xing:2007zj, Xing:2011ur}. For this purpose a modified version of the QR factorization is used, namely, the LQ decomposition, where $L$ corresponds to a lower triangular matrix and $Q$ is a unitary matrix. Moreover, in the context of the neutrino mixing, this lower triangular matrix is further split into the following form,
\begin{equation}
L=I-\alpha,
\end{equation}
where the matrix $\alpha$ is a lower triangular and describes a deviation from unitarity of the $U_{PMNS}$. 

Recently, a correspondence between the polar and QR parametrizations in the case of neutrino mixing was found \cite{Blennow:2016jkn}. 

In the end let us look briefly at the last factorization, i.e., Mostow decomposition. It decomposes any nonsingular complex matrix $A$ in the following way,
\begin{equation}
A=Ue^{iK}e^{S},
\end{equation}   
where $U$ is a unitary matrix, $K$ is a real skew symmetric matrix, and $S$ corresponds to a real symmetric matrix.

To this point we discussed matrix decompositions commonly used to parametrize a possible deviation from unitarity of the mixing matrix. Currently they are mostly used in top-down analyses \cite{Antusch:2006vwa,FernandezMartinez:2007ms,Xing:2007zj,Antusch:2009pm,Blennow:2011vn,Xing:2011ur,Escrihuela:2015wra,Blennow:2016jkn,Agostinho:2017wfs}, which means that they are considered as a part of a complete unitary matrix each time. As we have shown, such an approach trivially ensures the contraction property for these matrices. 
Let us note that top-down parametrizations are based on general treatment of unitarity breaking effects described by matrix factorization, and there is a lack of exact description based on entrywise parametrization of the mixing matrix, which would fulfill automatically the contraction property. Such a construction would be very useful. So far, parametrizations which are constructed fulfill a condition of contractions involving a general $I-\frac{\Theta \Theta^{\dag}}{2}$ representation of the matrix $\eta$, parametrizing a matrix $\Theta$ in such a way that $\frac{\Theta \Theta^{\dag}}{2}$ will fit into currently known limits on $\eta$. 

Actually in our strategy we come back to the bottom-up scenario,
as our analysis starts from the present state of knowledge on $\upmns$ mixing data in the form of an interval matrix, and we examine directly whether the matrices within are physically meaningful (i.e., are contractions). An extension of this idea allows us to define the complete region of physical mixing matrices as a convex hull of $\upmns$ matrices, which ensures that any physical mixing matrix can be constructed as a convex combination of $\upmns$ matrices. 
 
Now, let us emphasize the relation of our approach to the polar decomposition. In our analysis we use singular values as an indicator of whether a given matrix is a contraction. However, it is known that eigenvalues of the polar factor, which follows from the definition, are equal to singular values of an initial matrix. Thus from that perspective a polar decomposition can be treated as a compact version of singular value decomposition. Nevertheless, from a numerical analysis perspective, singular value decomposition algorithms are more natural, since they arise from the eigenvalue decomposition of matrices $AA^{\dag}$ and $A^{\dag}A$. Thus in most cases, in order to obtain an algorithm for a polar decomposition, we have to translate algorithms for the singular value decomposition.


%

\end{document}